\documentclass{article}
\usepackage{amsmath}
\usepackage{acronym}
\usepackage{amsthm}
\usepackage{pstool}
\usepackage{dsfont}
\usepackage{subfig}

\newtheorem{theorem}{Theorem}
\newtheorem{corollary}{Corollary}
\newtheorem{lemma}{Lemma}
\theoremstyle{definition}
\newtheorem{definition}{Definition}

\renewcommand{\vec}{\mathbf}
\newcommand{\set}[1]{\left\{#1\right\}}
\newcommand{\nb}{\text{nb}}
\newcommand{\abs}[1]{\left|#1\right|}

\newcommand{\R}{\mathds{R}}

\title{Perfectly Secure Communication, based on Graph-Topological Addressing in Unique-Neighborhood Networks}
\author{Stefan Rass\thanks{LIT Secure and Correct Systems Lab, Johannes Kepler University Linz, and Institute for Artificial Intelligence and Cybersecurity, Universitaet Klagenfurt, email: \texttt{stefan.rass@jku.at}}}

\begin{document}
\maketitle

\begin{abstract}
We consider network graphs $G=(V,E)$ in which adjacent nodes share common secrets. In this setting, certain techniques for perfect end-to-end security (in the sense of confidentiality, authenticity (implying integrity) and availability, i.e., CIA+) can be made applicable without end-to-end shared secrets and without computational intractability assumptions. To this end, we introduce and study the concept of a \emph{unique-neighborhood network}, in which nodes are uniquely identifiable upon their graph-topological neighborhood. While the concept is motivated by authentication, it may enjoy wider applicability as being a technology-agnostic (yet topology aware) form of addressing nodes in a network.
\end{abstract}

\section{Introduction}
Let a network be given as an undirected graph $G=(V,E)$, in which node adjacency $\set{u,v}\in E$ is characterized by two nodes $u$ and $v$ sharing a common secret (key). Consider the following question:
\begin{quote}
  Can any two nodes $u,v\in V$, which are not adjacent (i.e., $\set{u,v}\notin E$) exchange messages with computationally unconditional privacy, authenticity and reliability?
\end{quote}
By its formulation, an answer cannot use any complexity-theoretic intractability assumptions, thus ruling out public-key cryptographic techniques. Availability is typically a matter of redundancy, and it is known that both, unconditional privacy and availability are both achievable by certain graph connectivity properties and \ac{MPT} based on secret sharing \cite{Wang&Desmedt2008}. In the simplest yet provably most efficient setting \cite{Fitzi&Franklin2007}, Alice sends a message $m$ to Bob according to the following scheme:
\begin{enumerate}
  \item Alice encodes $m$ into a set $s_1,\ldots,s_k$ shares, so that any $d$ out of these $k$ shares suffice to reconstruct the message, while no set of less than $d$ shares leaks any information about $m$. This is a standard application of polynomial $(d,k)$-threshold secret sharing \cite{Shamir1979}.
  \item Alice chooses a set of $k$ node-disjoint paths $\pi_1,\ldots,\pi_k$ from her node to Bob. Herein, two paths $\pi_i,\pi_j$ are said to be \emph{node-disjoint}, if they satisfy $V(\pi_1)\cap V(\pi_2)=\{$Alice, Bob$\}$, where $V(\pi)$ is the vertex set of the path. That is, any two paths intersect nowhere in $G$ except at the end-points.
  \item Bob reconstructs the secret as usual for polynomial secret sharing, potentially recovering from up to $\lfloor(k-d)/2\rfloor$ errors. This recovery is possible by the Welch-Berlekamp algorithm \cite{Berlekamp&Welch1986}, exploiting the known ``isomorphy'' between polynomial secret sharing and Reed-Solomon encoding \cite{McElice&Sarwate1981}.
\end{enumerate}
This form of \ac{MPT} achieves confidentiality against any attacker being able to sniff on $<d$ nodes, by design of the secret sharing. Reliability of the transmission follows from the error correction capability of the sharing treated as an error-correcting code. Note that this protocol does not need any point-to-point encryption, if the attacker is constrained to eavesdrop on nodes only, and on strictly less than $d$ of them. This is what we shall assume w.l.o.g. throughout the paper\footnote{The case of an attacker being able to listen on all the lines in a network requires point-to-point encryption, which, for unconditional security, would call for additional techniques like quantum key distribution protocols (e.g., BB84 \cite{Bennett&Brassard1984}) in the network. We leave such technological extensions aside in this work, and shall remain independent of any such assumptions in stating that shared secrets just ``exist'', without adopting any prescriptions on how this is practically done.}.

The obstacle towards practical implementations of this scheme is the network
needing to provide $k$ node-disjoint paths. This property is, by a Theorem of
H.Whitney \cite[Thm. 5.17]{Chartrand&Zhang2005} equivalent to
$k$-vertex-connectivity of $G$: a graph $G=(V,E)$ is $k$-connected (more
specifically $k$-vertex-connected), if it takes at least $k$ nodes to be
removed from $G$ until the graph becomes disconnected. That is, Alice and Bob
cannot be disconnected by removing up to any $k-1$ nodes between them.
Whitney's theorem equates this condition to the existence of $k$
node-disjoint paths between Alice and Bob. While this is a strong
connectivity requirement in general, 2-connectivity is a highly common
feature of networks for the sake of resilience against single node failures.
The general problem of extending a graph into $k$-connectedness is
computationally intractable \cite{Rass2014}, but the construction of large
$k$-connected graphs from smaller ones is inductively easy \cite{Rass2009b}:
\begin{enumerate}
  \item Start with the $\subseteq$-smallest\footnote{Herein, the \emph{subgraph} relation $G_1\subseteq G_2$ between two graphs $G_1=(V_1,E_1), G_2=(V_2,E_2)$ holds if $V_1\subseteq V_2$ and $E_1\subseteq E_2$.} $k$-connected graph, which is the complete graph $K_{k+1}$.
  \item Given any two $k$-connected graphs $G_1=(V_1,E_1), G_2=(V_2,E_2)$, pick $k$ nodes from each graph, denoted as $u_1,\ldots,u_k\in V_1$ and $v_1,\ldots,v_k\in G_2$, and form the connected graph $H=(V_1\cup V_2$, $E_1\cup E_2$ $\cup \{(u_i,v_i):$ $i=1, \ldots, k\})$. Then $H$ is again $k$-connected.
\end{enumerate}
Thus, if the network is to be constructed from scratch in a hierarchical fashion, maintaining it $k$-connected is a simple matter of proper connections between subnetworks into the bigger network.

Authenticity is a different story, but achievable along similar lines: \ac{MPA} resembles the common form of how hand-written signatures are verified in companies. Departments typically maintain samples of the handwritten signature of a decision maker to verify it on paper documents (otherwise, anyone could just scrawl some name and claim it to be someone else's handwritten signature). The digital version of this procedure uses the point-to-point shared secrets to mimic a ``signature'' by a conventional \ac{MAC}. Specifically, if Alice's node $v$ has a neighbor set $\nb(v)=\set{w\in V: \set{v,w}\in E}$, and shares a key $s_i$ with each neighbor $i\in\nb(v)$. She can use these to undersign a message using a set of \acp{MAC} under the keys $s_1,\ldots,s_k$ for $k\leq\abs{\nb(\text{Alice})}$. Bob, upon reception of the \acp{MAC}, can ask Alice's neighbors for verification, and properly react upon their replies (along node-disjoint paths again). This protocol is depicted in Figure \ref{fig:mpa} and, using techniques of game theory and universal hashing for the functions $MAC$ and $h$, is provably secure without any computational intractability assumptions \cite{Rass&Schartner2010}. It must be noted that the vertex-connectivity number $k$ of the graph needs not be equal to the number of neighbors or paths used; it must only be large enough to admit the sought number of neighbors/paths. As such, the number (here $k$) of paths can be less than the vertex-connectivity number of the graph (also denoted as $k$ here).

\begin{figure}
  \centering
  \includegraphics[width=0.7\textwidth]{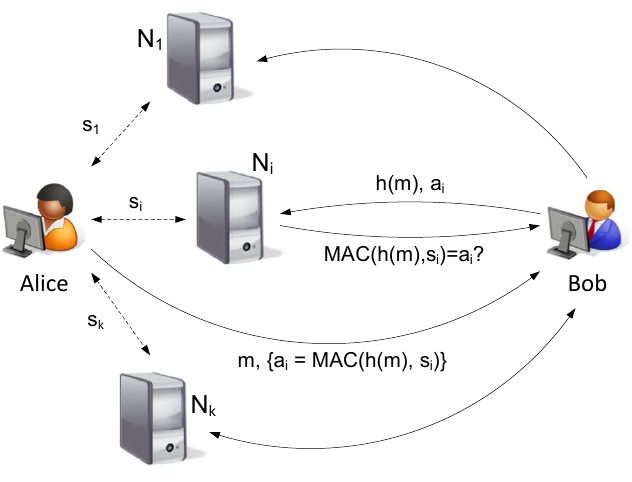}
  \caption{Multipath authentication -- Example scheme}\label{fig:mpa}
\end{figure}

The remaining question concerns the uniqueness of Bob in being the ``signer''
of the message in the \ac{MPA} scheme. Is it only Alice that could have
attached the \ac{MAC} set $\set{MAC(h(m), s_i): i\in\nb(\text{Alice})}$?
Apparently so, if the neighborhood of Alice \emph{uniquely} characterizes,
resp. distinguishes, her from all other nodes in the network. This, however,
is a nontrivial property of a graph, and in the center of study in this work
hereafter.

\section{Problem Statement}
We study the problem of characterizing a node $v\in V$ based only on its graph-topological neighborhood, i.e., we are interested in graphs with the following property:
\begin{definition}[Unique-Neighborhood Network]
A graph $G=(V,E)$ is a \emph{\ac{wUNN}}, if the mapping $v\in V\mapsto
\nb(v):=\set{w: \set{v,w}\in E}$ is injective. We call it a \emph{(strong)
\ac{UNN}}, if no neighborhood is a subset of another node's neighborhood.
That is, for every $v$, there is some $u\in\nb(v)$ with $u\notin\nb(w)$ for
all $w\neq v$.
\end{definition}

The existence of such graphs is immediate by simple examples, such as lines
(Figure \ref{fig:unn-line}), circles, or the complete graph. The property,
however, may arise or vanish upon adding edges. For instance, the graph in
Figure \ref{fig:wunn} is a \ac{wUNN}, but loses this property upon adding the
edge $\set{2,4}$ to it, as in Figure \ref{fig:non-unn}. It regains
unique-neighborhoods, however, when the edges $\set{1,3}$ and $\set{2,4}$ are
added (Figure \ref{fig:unn-complete-graph}).

The distinction between weak and (strong) \ac{UNN} is necessary because the
latter lend themselves better to authentication matters: looking at the graph
in Figure \ref{fig:wunn}, node 2 has neighbor set $\set{1}$ which is
contained in the neighborhood $\set{1,3}$ of node 4 as well, so node 4 could
use a subset of its neighbors to mimic being node 2, based on neighbor sets
(only). Likewise, node 1 could pretend being node 3, because it has, among
others, also the neighbors that node 3 knows, and could use those nodes for
impersonating 3. Node 3, in turn, could not do this, as long as 1 uses its
full neighbor set of authentication. This possibility vanishes if no
neighborhood is a strict subset of another neighborhood.

\begin{figure}[b!]
  \centering
  \subfloat[\ac{UNN}: Line graph]{\phantom{XXXXXXXXXXX}\includegraphics[scale=1]{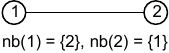}\label{fig:unn-line}\phantom{XXXXXXXXXXX}}

  \subfloat[\ac{UNN}: Complete graph $K_4$]{\phantom{XXXXXXXXXXX}\includegraphics[scale=1]{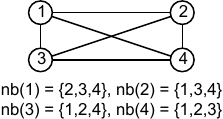}\label{fig:unn-complete-graph}\phantom{XXXXXXXXXXX}}

  \subfloat[A weak \ac{UNN}]{\phantom{XXXXXXXXXXX}\includegraphics[scale=1]{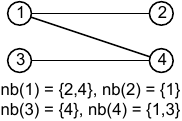}\label{fig:wunn}\phantom{XXXXXXXXXXX}}

  \subfloat[Not a \ac{UNN}: the complete bipartite graph $K_{2,2}$]{\phantom{XXXXXXXXXXX}\includegraphics[scale=1]{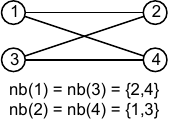}\label{fig:non-unn}\phantom{XXXXXXXXXXX}}

  \caption{Examples}\label{fig:examples}
\end{figure}

\section{An Sufficient Algebraic Condition for \ac{UNN}}\label{sec:algebraic-conditions}

Hereafter, let $n=\abs{V}$ be the number of nodes in $G$. Let $\vec
A\in\set{0,1}^{n\times n}$ be the graph's adjacency matrix, with $a_{ij}=1$
whenever node $i$ is connected to node $j$ and zero otherwise. Each
row/column of the matrix thus corresponds to a node, and the $i$-th row in
$\vec A$ can be taken as a vector of indicators, describing the neighbourhood
of node $i$. Likewise, since $G$ is undirected and $\vec A$ is hence
symmetric, the same goes for the columns of $\vec A$.

Furthermore, since we are considering \ac{UNN} for matters of authentication,
we assume the graph to have no loops (which would correspond to a party
self-certifying the validity of its own \ac{MAC}), no multi-edges (as the
connection as such counts, not how many cables connect two instances), and to
be connected (as isolated nodes could not communicate and hence have no need
for authentication). Thus, we hereafter consider graphs that are
\emph{simple}, i.e., without loops or multi-edges, and connected. Since we
also consider lines for bidirectional communication (necessarily for the
protocol above to work), we do not consider directed graphs hereafter.

A \ac{UNN} has pairwise distinct rows/columns, so it is conceptually
straightforward to sort the rows of $\vec A$ in any order, and to look for
adjacent identical rows. While the problem is algorithmically easy, we can
also give an algebraic condition to induce the \ac{UNN} property. Several
sufficient conditions are immediate to imply distinct rows, such as full rank
of the adjacency matrix $\vec A$, or the matrix being orthogonal. Those,
however, is also necessary, so these are overly strong for our purposes.

In a more direct approach, we can directly ask for pairwise distinctness, and
it turns out that this is easy to cast into an algebraic condition: let $\vec
x_i$ be the $i$-th row of the adjacency matrix $\vec A$. Consider the $ij$-th
entry $a_{ij}^{(2)}$ in $\vec A\cdot \vec A^T=\vec A^2$ (the bracketed
superscript in $a_{ij}^{(2)}$ shall be a reminder that we indeed do not just
square $a_{ij}$, since actually $a_{ij}^{(2)}=\sum_{k=1}^n a_{ik} a_{kj}=\vec
x_i^T \vec x_j$).

If $i\neq j$, then two cases are possible:
\begin{enumerate}
\item if $\vec x_i=\vec x_j$, then we can write $a_{ij}^{(2)}=\vec x_i^T \vec x_i$
\item Otherwise, if $\vec x_i\neq \vec x_j$, then the sum giving $\vec
    x_i^T \vec x_j$ will be such that at least one 1-entry in $\vec x_i$
    ``matches'' with a 0-entry in $\vec x_j$ (or vice versa), since
    otherwise, the two vectors would be identical. Thus, the sum $\vec
    x_i^T \vec x_j$ must be strictly less than $\vec x_i^T \vec x_i$ by at
    least 1, since by the strong \ac{UNN} condition, there must be a
    neighbor $k$ of node $i$ that node $j$ does not have, hence the 1-entry
    at position $k$ in $\vec x_i$ hits a zero entry in the same position in
    $\vec x_j$, and vice versa. Hence, $\vec x_i^T\vec x_j<\|\vec x_i\|_1$
    and also $\vec x_i^T \vec x_j<\|\vec x_j\|_1$.
\end{enumerate}
Now, look into the first row of $\vec A^2=\vec A^T \vec A=\vec A\vec A^T$: in
this row, we have all scalar products $\vec x_1^T \vec x_i$, for which the
result is always $\leq \vec x_1^T \vec x_1$, with equality if and only if
$\vec x_i=\vec x_1$. That is, we can test for another copy of $\vec x_1$ to
exist by looking if all entries are less than the first. This is already half
the test: the matrix $\vec A$ has pairwise distinct rows if and only if the
diagonal element in the $i$-th row of $\vec A^2$ is strictly greater than all
other elements in that row (likewise, column).

We can convert this into an algebraic inequality condition: write
$\mathbf{1}_{n\times m}$ to mean the $(n\times m)$-matrix of all 1-entries,
and observe that $\vec x_i^T \vec x_i=\vec x_i^T \mathbf{1}_{n\times 1}$.
Thus, we can create the diagonal entry in $\vec A^2$ alternatively by
multiplying with a vector of all 1es. Extending this idea, note that the
entry $\vec x_i^T \vec x_i$ appears in all columns along the $i$-th row in
$\vec A\cdot \mathbf{1}_{n\times n}$. Our condition asks for the diagonal to
be strictly larger than the other elements in the same row, and all we have
left to do ``artificially increase'' the actual diagonal by 1 so that the
$<$-condition holds there too (otherwise we would have equality).

Overall, we end up with the following sufficient condition:
\begin{theorem} A simple graph $G=(V,E)$ with $n=\abs{V}$ nodes is a \ac{UNN}, if its adjacency matrix $\vec A$ satisfies
\begin{equation}\label{eqn:unn-condition}
  \vec A\cdot \mathbf{1}_{n\times n}+ \vec I_n-\vec A^2\geq \mathbf{1}_{n\times n}.
\end{equation}
where $\vec I_n$ is the $n$-th identity matrix and the inequality holds per element.
\end{theorem}


A similar condition can be obtained for weak \ac{UNN}, only considering that
some node $i$ may have a neighbor set that covers that of another node $j$.
Like above, consider the product $\vec x_i\cdot\vec x_j$, and let
$\nb(i)\subset\nb(j)$, then the 1es in $\vec x_i$ all match with 1es in $\vec
x_j$, except for a few more 1-entries that $\vec x_j$ may have. Hence, $\vec
x_i^T\vec x_j= \vec x_i^T\vec x_i$, since all of $i$'s neighbors go into the
count (as $\vec x_j$ has a 1-entry for every 1-entry in $\vec x_i$).
Conversely, $\vec x_j$ will have a neighbor $k$ that is not counted upon
multiplying with $\vec x_i$, since $k\notin\nb(i)$, and therefore we have
$\vec x_i^T\vec x_j<\vec x_j^T\vec x_j$. The two neighborhoods are thus
distinct if and only if at least one of the two inequalities is strict, or by
adding them, $2\vec x_i^T\vec x_j<\vec x_i^T\vec x_i+\vec x_j^T\vec x_j$, and
equality is only possible if the two vectors are identical, as in that case,
$\nb(i)\subseteq\nb(j)$ and $\nb(j)\subseteq\nb(i)$ so that the neighborhood
is the same. 


%

Constructing optimal networks w.r.t. some cost function $c:\set{0,1}^{n\times
n}\to\R$ on the adjacency matrix is then a matter of constrained nonlinear
optimization:
\begin{align}
    \left.\begin{array}{ll}
    \min_{\vec A\in\set{0,1}^{n\times n}}&c(\vec A)\\
    \text{subject to }&\eqref{eqn:unn-condition}.
    \end{array}\right\}\label{eqn:optimization}
\end{align}

One instance of
\eqref{eqn:optimization} could be, for example, looking for the smallest
\ac{UNN} in which a given graph $G$ appears as a subgraph. Conversely, we can
look for the largest subgraph inside $G$ that is a \ac{UNN}. This is done in
the next section.

\section{Construction and Encounter of \acp{UNN}}

It is highly unlikely that random graphs, e.g., scale-free or others come up
as \acp{UNN}. Taking the internet as an example of a scale-free topology,
just consider an \ac{ISP} with a set of customers. Each customer is typically
connected to only one \ac{ISP}, making this node the only and hence
non-unique neighborhood of the customer. However, an \acp{ISP} located in the
center of a star topology may consider its customers as a unique neighborhood
to another \ac{ISP}, unless the two have identical sets of customers.

Though a graph $G=(V,E)$ may not be a \ac{UNN}, is there perhaps a subgraph
$G'=(V',E')$ with $V'\subseteq V, E'\subseteq E$ that is a \ac{UNN}? A
positive answer is reached by looking at trees first. First, let us define
the \emph{degree} of a node $v$ in an undirected graph is the number
$\abs{\nb(v)}$.
\begin{lemma}\label{lem:trees-have-unn-subgraphs}
Let $G=(V,E)$ be a tree, and let $S$ be the set of all nodes having degree 1 in $G$. Then, all nodes in $V\setminus S$ have unique neighborhoods in (the full graph) $G$.
\end{lemma}
\begin{proof}
The nodes in $S$ are all leafs in the tree, and their exclusion leaves the
mapping $f:v\mapsto \nb(v)$ restricted to inner nodes only. Consequently, we
consider injectivity of $f$ only on the set $V\setminus S$, but with the
neighborhood $\nb$ being determined by the full set $V$. For any two inner
nodes $u,v$, those may have the same parents, but since $G$ is a tree, they
have distinct children, thus making the corresponding neighborhood (composed
from parents and children, possibly also from $S$) pairwise distinct.
\end{proof}

Lemma \ref{lem:trees-have-unn-subgraphs} is not optimal in the sense that a leaf that is singleton (in the sense of having no ``siblings'', i.e., sharing its parent with no other vertex in $G$) can be included, and the graph remains a \ac{UNN}. More generally, we can even connect any two disjoint \acp{UNN} by a single or multiple edges, with the resulting graph again being a \ac{UNN}.

\begin{lemma}\label{lem:unn-connection}
Let $G_1=(V_1,E_1), G_2=(V_2,E_2)$ be \acp{UNN}, with $V_1\cap
V_2=\emptyset$. Select any two distinct nodes $u\in V_1,v\in V_2$ and
construct the graph $H=(V_1\cup V_2, E_1\cup E_2\cup\set{\set{u,v}})$. Then,
$H$ is a \ac{UNN}.
\end{lemma}
\begin{proof}
Since we connect a node $u\in V_1$ to $v\in V_2$ (only), the neighborhoods of
both are just extended by the other. But since neither node appears in the
other graph, the neighborhoods remain distinct. The remaining nodes that were
not involved in the connection between $G_1, G_2$ retain their neighborhoods
as they were, which are again unique since the graphs were vertex-disjoint.
\end{proof}
Unfortunately, Lemma \ref{lem:unn-connection} does not lend itself to a
greedy construction algorithm for a node- or edge-minimal \ac{UNN}, since the
property is non-monotone in terms of the subgraph relation, as the example
graphs in Figure \ref{fig:examples} show. Thus, the set of all \acp{UNN} does
not form a matroid. Given two \acp{UNN}, we can connect them with a single
edge to form one graph that is a \ac{UNN}, and from that point onwards add
edges to $H$, as long as condition \eqref{eqn:unn-condition} on the adjacency
matrix tells that we retain a \ac{UNN}.

More interesting for our purposes is the corollary from these results:
\begin{corollary}\label{cor:unn-subgraph}
Every undirected graph $G=(V,E)$ has a vertex-maximal subgraph $G'=(V',E')$ that is a \ac{UNN}, and which only excludes nodes (if any) of degree 1 in $G$.
\end{corollary}
\begin{proof}
Take $T=(V_T, E_T)$ as any spanning tree in $G$, then Lemma \ref{lem:trees-have-unn-subgraphs} tells that every inner node of $T$ has unique neighborhoods. Let $V_T^o$ be the set of inner nodes in $T$, and let $E'=(V_T^o\times V_T^o)\cap E$ be the induced edge set. Let $v\in V_T^o$ be an arbitrary node for which all $u\in\nb(v)$ have degree 1. That is, $v$ is a(ny) node that is directly connected to leaf nodes in $T$. For each such node $v$, we can pick exactly one of its children $c\in V_T$ (a leaf) arbitrarily, and add the edge $\set{v,c}$ to $E'$. By Lemma \ref{lem:unn-connection}, the so-extended tree remains a \ac{UNN}. Moreover, it is a vertex-maximal such subgraph, since adding any further leaf node connected to some $v\in V_T^o$, we would end up with two children $c_1, c_2$ of $v$ whose common neighborhood is $\set{v}$. Since $T$ was spanning, there are no other nodes that we could add.
\end{proof}
Obviously by construction, Corollary \ref{cor:unn-subgraph} only assures a
\ac{UNN} subgraph of ``minimal'' vertex-connectivity as being a tree. To
construct a $k$-connected \ac{UNN}, Corollary \ref{cor:unn-subgraph} is
apparently not very useful as the inner \ac{UNN} is only 1-connected. Its
primary applicability is rather to provide a starting point for optimization
like in \eqref{eqn:optimization}: if $G$ is not a \ac{UNN}, we can pick a
maximal subgraph of $G$ that is a \ac{UNN}, and keep adding edges (from $G$
or new ones), until $G$ has been extended into the ``smallest'' \ac{UNN} that
covers $G$.

\section{Graph-Topological Addressing and Security}\label{sec:taa}
The existence of a vertex-maximal subgraph being a \ac{UNN} is not at all surprising, but the important fact is that the network is a \ac{UNN} until the ``last mile'' to the customer, who typically is a node of degree 1. Routing messages into such nodes, however, makes sense only if the node itself is already the receiver, so the actual addressing is only needed until the hop right before the final node. Once this butlast node has been reached, the final point-to-point connection needs no further addressing.

Corollary \ref{cor:unn-subgraph} is thus the final key to ``technology
agnostic addressing'' in the sense as we look for: let $G=(V,E)$ be any
network that is not necessarily a \ac{UNN}. Within $G$, we can construct a
minimal spanning tree (by known algorithms), which relative to the entire
network $G$ is already a \ac{UNN}. All nodes outside this \ac{UNN} are
excluded only for sharing a common connection point into the \ac{UNN} $T$.
Now, consider the nodes in $T$ as \acp{ISP}, arranged in some hierarchical
structure that the tree reflects, then the nodes outside $T$ are all
customers connected to the same \ac{ISP}, but these are directly reachable
from their individual \ac{ISP} node. Thus, the only nodes with a non-unique
neighborhoods are the \acp{ISP}'s customers, while within the larger
(inter)net, nodes can be addressed purely using their localities.

The addressing of nodes based on their neighborhoods is \emph{agnostic} of technology, but in the same degree needs to be \emph{aware} of topology. As such it may not be applicable in certain specialized domains such as on-chip networks or ad hoc networks. The main area of application are hence fixed network installations such as those maintained by \acp{ISP}.

This technological/graph-topological effort comes with a significant practical advantage from the perspective of \emph{usable security}: it requires key-management in the sense of exchanging common secrets, only between a considerably small number of nodes. Precisely, while symmetric end-to-end encryption in a network of $n$ nodes would require $O(n^2)$ keys to be exchanged, multipath schemes as described here require only $\abs{E}$ such keys to be shared between direct neighbors, while still providing end-to-end security without computational intractability. More important from a practical perspective is the fact that the key-management does not need to rest with Alice or Bob as users: unlike public-key encryption that relies on complex certificate management that to a wide extent runs on the application layer (and hence is in Alice and Bob's direct hands), multipath transmission and \ac{UNN} maintenance (i.e., establishment, broadcast and updates to neighborhoods) can run entirely below the application layer, thus providing addressing and confidentiality in a technology agnostic form and completely transparent for all users. Especially matters of maintaining a \ac{UNN} are flexible and do not need to take into account the entire network topology: an implicit point made in the proof of Corollary \ref{cor:unn-subgraph} is the fact that we \emph{do not} need all graph-topological neighbors, while it suffices to include only a selection of them to define unique neighborhoods. This also extends to the key-management: for \ac{MPA}, it is only necessary to share keys with the neighbors relevant for the addressing, but not with all neighbors that may physically exist. Thus, the lot of key material maintained by the network additionally shrinks, since the ``inner'' \ac{UNN} is merely a spanning tree inside the actual network, on which shared keys are required. Since a tree with $n$ nodes has $m=n-1$ edges, we end up with $O(n)$ keys necessary for confidential and authentic end-to-end security, as opposed to $O(n^2)$ in the conventional setting of symmetric encryption.

\section{Discussion}
Multipath transmission schemes in \acp{UNN} offer the remarkable possibility of doing something like public-key cryptography without public-key cryptography: the sender Alice can send an authentic and confidential message to Bob, with both properties provably implied by graph-topological features only, and not resting on (unproven) computational intractability. Moreover, Alice and Bob use only publicly available information for that purpose, with all matters of symmetric cryptography being handled below the application layer (thus, the crypto is entirely transparent). Unlike computational intractability, an \ac{ISP} can easily assure these conditions to hold in a publicly verifiable manner, and without bothering its customers Alice or Bob with any key-management at all. This may especially be interesting in emerging networks such like the \ac{IoT}.

This note leaves a set of problems open for future research, such as properties and complexity of the optimization problem \eqref{eqn:optimization}, and practical (algorithmic) matters of routing. For ad hoc networks, it is interesting to compute the likelihood of a random graph to be a \ac{UNN}. This property will not necessarily emerge suddenly unlike other graph properties, due to lack of monotonicity. However, the non-monotonicity in connection with the application in the \ac{IoT} and ad hoc networking may render \acp{UNN} interesting objects to study in the context of random graphs.

\subsection*{Acknowledgment}
The author thanks Julian Gaggl for pointing out some corrections in earlier
versions of this draft.

\bibliographystyle{abbrv}

\begin{acronym}
\acro{UNN}{Unique-neighborhood network}%
\acro{wUNN}{weak Unique-neighborhood network}%
\acro{ISP}{Internet service provider}%
\acro{MPT}{Multipath transmission}%
\acro{MPA}{Multipath authentication}%
\acro{MAC}{Message authentication code}%
\acro{IoT}{internet-of-things}%
\end{acronym}

\end{document}